\pdfoutput=1
\RequirePackage{ifpdf}
\ifpdf 
\documentclass[pdftex]{sigma}
\else
\documentclass{sigma}
\fi

\numberwithin{equation}{section}

\def\HH{\mathcal{H}}
\def\FF{\mathcal{F}}
\def\AA{\mathcal{A}}

{ \theoremstyle{definition}
\newtheorem*{rem}{Remark}}

\begin{document}
\allowdisplaybreaks

\newcommand{\arXivNumber}{1904.03993}

\renewcommand{\PaperNumber}{082}

\FirstPageHeading

\ShortArticleName{One Parameter Family of Jordanian Twists}

\ArticleName{One Parameter Family of Jordanian Twists}

\AuthorNameForHeading{D.~Meljanac, S.~Meljanac, Z.~\v{S}koda and R.~\v{S}trajn}

\Author{Daniel MELJANAC~$^{\dag^1}$, Stjepan MELJANAC~$^{\dag^2}$, Zoran \v{S}KODA~$^{\dag^3}$ and Rina \v{S}TRAJN~$^{^{\dag^4}}$}

\Address{$^{\dag^1}$~Division of Materials Physics, Institute Rudjer Bo\v{s}kovi\'c,\\
\hphantom{$^{\dag^1}$}~Bijeni\v{c}ka cesta~54, P.O.~Box 180, HR-10002 Zagreb, Croatia}
\EmailDD{\href{mailto:Daniel.Meljanac@irb.hr}{Daniel.Meljanac@irb.hr}}

\Address{$^{\dag^2}$~Theoretical Physics Division, Institute Rudjer Bo\v{s}kovi\'c,\\
\hphantom{$^{\dag^2}$}~Bijeni\v{c}ka cesta~54, P.O.~Box 180, HR-10002 Zagreb, Croatia}
\EmailDD{\href{mailto:meljanac@irb.hr}{meljanac@irb.hr}}

\Address{$^{\dag^3}$~Department of Teachers' Education, University of Zadar,\\
\hphantom{$^{\dag^3}$}~Franje Tudjmana 24, 23000 Zadar, Croatia}
\EmailDD{\href{mailto:zskoda@unizd.hr}{zskoda@unizd.hr}}

\Address{$^{\dag^4}$~Department of Electrical Engineering and Computing, University of Dubrovnik,\\
\hphantom{$^{\dag^4}$}~\'Cira Cari\'ca 4, 20000 Dubrovnik, Croatia}
\EmailDD{\href{mailto:rina.strajn@unidu.hr}{rina.strajn@unidu.hr}}

\ArticleDates{Received April 16, 2019, in final form October 19, 2019; Published online October 25, 2019}

\Abstract{We propose an explicit generalization of the Jordanian twist proposed in $r$-symmetrized form by Giaquinto and Zhang. It is proved that this generalization satisfies the 2-cocycle condition. We present explicit formulas for the corresponding star product and twisted coproduct. Finally, we show that our generalization coincides with the twist obtained from the simple Jordanian twist by twisting by a 1-cochain.}

\Classification{53D55; 16T05}

\section{Introduction}

Drinfeld twists of Hopf algebras~\cite{drinfeld} provide a systematic way of producing new examples in noncommutative geometry. Given a Hopf algebra $\HH$ with a coproduct $\Delta_0$, a counit $\epsilon_0$, and an element $\FF\in \HH\otimes\HH$ satisfying a 2-cocycle condition and a normalization (counitality) condition~\cite{chari, etingofKazh,etingofSchiff,majid}, often called a Drinfeld twist, one defines a new coproduct, $\Delta^{\mathcal{F}}(-) = \mathcal{F}^{-1}\Delta_0(-)\mathcal{F}$, which is coassociative due to the 2-cocycle condition. Moreover, $\HH$~as an algebra, together with the new coproduct $\Delta^{\mathcal{F}}$ becomes a new {\em twisted} Hopf algebra $\HH^{\mathcal{F}}$. Along with a Hopf algebra, many associated constructions like its representations, comodules, module algebras and so on, are twisted as well, using standard formulas involving the twist~$\mathcal{F}$. The systematic nature of the twisting procedure makes it suitable for finding new physical models with the Hopf algebra covariance built in.

In 1989, a new construction of a deformation 2-cocycle is proposed by Coll, Gerstenhaber and Giaquinto in~\cite{CollGerstGiaq1989}. Their construction starts with a $k$-algebra $A$ with multiplication $\mu_A$ and two derivations $\phi,\psi\colon A\to A$ satisfying $[\phi,\psi] = \lambda\psi$ for some $\lambda\in k$. The action of the 2-dimensional Lie algebra $L$ generated by $\phi$ and $\psi$ on $A$ by derivations extends to a unique action $\triangleright$ of the Hopf algebra $U(L)$ on $A$ making it a $U(L)$-module algebra. They prove~\cite{CollGerstGiaq1989,CollGGGrabeVlassov98} that $\mu_A\circ(\phi\otimes\psi)$ is a~Hochschild 2-cocycle which may be integrated to yield a formal deformation of $A$. The deformed multiplication on $A[[t]]$ is given by $\mu_A\circ(1\otimes 1+t\phi\otimes 1)^{1\otimes\psi}\!\circ(\triangleright\otimes\triangleright)$ for $\lambda=1$. This formula involves the element
\begin{gather*}
 (1\otimes 1+t\phi\otimes 1)^{1\otimes\psi} = \sum_{n = 0}^\infty t^n \phi^n\otimes \binom{\psi}{n} \\
 \hphantom{(1\otimes 1+t\phi\otimes 1)^{1\otimes\psi}}{} = \sum_{n=0}^\infty \frac{t^n}{n!}\phi^n\otimes\psi(\psi-1)\cdots(\psi-n+1)\in(U(L)\otimes U(L))[[t]],
\end{gather*}
which is a (Drinfeld) 2-cocycle twist for the Hopf algebra $U(L)[[t]]$. They provide several examples. Their construction is reanalyzed in detail in~\cite{CollGGGrabeVlassov98} and the 2-cocycle twist has been rediscovered in~\cite{Og}. Algebras obtained by variants of their method are now often called Jordanian deformations. Most studied examples are a Jordanian deformation of the universal enveloping algebra $U(\mathfrak{sl}(2))$ (and its dual) and more general Jordanian quantum groups, leading to the corresponding Jordanian classical $r$-matrices and quantum $R$-matrices (some of which were known before, e.g.,~\cite[Example~1, due D.~Gurevich]{lyub} and \cite[Section~2.2]{demidov}). Regarding that $U(\mathfrak{sl}(2))$ can be embedded into Yangian $Y(\mathfrak{sl}(2))$, it is not surprising that more elaborate versions of Jordanian twists are used to obtain new deformations of Yangians~\cite{stolinYsl2,stolinY2, stolinY1}, with applications to integrable models, chain models in particular~\cite{stolin2001}. A comprehensive study of a related class of classical $r$-matrices can be found in~\cite{stolin2008}.

Here we present another approach to new Jordanian deformations. Closer to the setup of our paper, consider the universal enveloping algebra of the 2-dimensional solvable Lie algebra with generators $H$ and $E$ with $[H,E] = E$. Define
\begin{gather*}
H^{\langle m\rangle }= H(H+1)\cdots(H+m-1), \qquad H^{\langle 0\rangle }=1.
\end{gather*}
Giaquinto and Zhang in~\cite[Theorem~2.20]{GZ}\footnote{The twist $\mathcal{F}$ in~\cite{GZ} is renamed here as $\mathcal{F}^{-1}$.} proposed the Jordanian twist~\cite{Og} in $r$-symmetrized form
\begin{gather*}
\mathcal{F}_{{\rm GZ}}^{-1} = \sum_{m=0}^{\infty} \frac{t^m}{m!} \sum_{r=0}^m (-1)^r \binom{m}{r}E^{m-r}H^{\langle r\rangle }\otimes E^r H^{\langle m-r\rangle}.
\end{gather*}

This twist can also be written as
\begin{gather*}
\mathcal{F}_{{\rm GZ}}^{-1} = \sum_{m=0}^\infty t^m \sum_{r=0}^\infty (-E)^{m-r}\binom{-H}{r} \otimes E^r \binom{-H}{m-r}
 = \sum_{k,l=0}^\infty t^{k+l}(-E)^k \binom{-H}{l}\otimes E^l \binom{-H}{k}.
\end{gather*}
We shall use a different notation in this paper, namely
\begin{gather*}
E=P, \qquad H=-D, \qquad [D,P]=-P.
\end{gather*}
This suggests an interpretation of $D$ as the relativistic dilation operator and $P$ as the momentum in some applications. We introduce a family of twists $\mathcal{F}_{{\rm GZ},u}^{-1}$, parametrized by parameter $u$, via an explicit series~(\ref{eq:FGZu}). This family interpolates between the Jordanian twists $\mathcal{F}^{-1}_0$ and $\mathcal{F}^{-1}_1$, where
\begin{gather}\label{eq:F0F1}
 \mathcal{F}_0= \exp \left( -\ln\left(1-\frac{1}{\kappa}P\right) \otimes D\right)\qquad \mbox{and}\qquad \mathcal{F}_1= \exp \left( -D\otimes \ln\left( 1+\frac{1}{\kappa}P\right)\right).
 \end{gather}
Our main interest in Jordanian twists is due to their appearance~\cite{BP, tolstoy1} in the study of $\kappa$-deformed Minkowski space (where the intepretation of $D$ and $P$ as the dilation and momentum operators also makes sense), where $\kappa$ is viewed as being linked to the scale of quantum gravity~\cite{LukRuegg,LukTol, tolstoy2}.

Any Drinfeld twist $\mathcal{F}$ can be modified by any 1-cochain $\omega\in\HH$, producing a new twist $\big(\omega^{-1}\otimes\omega^{-1}\big) \mathcal{F}\Delta(\omega)$, see~\cite{majid}. In an earlier paper~\cite{cobtw}, this procedure has been used to obtain a~certain twist $\mathcal{F}_{R,u}^{-1}$ for every $u$. In that context, it has been written in the form of a product of three exponential factors, see also reference~\cite{MMPP}. Regarding that it is obtained from a 2-cocycle by modification by a 1-cochain implies that it is itself a 2-cocycle.

Twists $\mathcal{F}_{{\rm GZ},u}$ and $\mathcal{F}_{R,u}$ generate the same Hopf algebra. It is proved in this paper that our generalized Giaquinto--Zhang twist $\mathcal{F}_{{\rm GZ},u}^{-1}$ satisfies the same differential identity as $\mathcal{F}_{R,u}^{-1}$, including the initial condition; consequently the two twists coincide. The importance of this result is that while $\mathcal{F}_{{\rm GZ},u}^{-1}$ is introduced via an explicit series expansion more suited for other calculations, the very construction of $\mathcal{F}_{R,u}^{-1}$ ensures that it is a 2-cocycle; we however also exhibit an elaborate proof of the 2-cocycle condition, directly from the definition of $\mathcal{F}_{{\rm GZ},u}^{-1}$.

The exposition is organized as follows. In Section~\ref{sec:GZ}, we define an interpolation $\mathcal{F}_{{\rm GZ},u}^{-1}$ via an explicit expansion and show that it has the claimed limits at $u=0$ and $u=1$. In Section~\ref{ssec:coc}, we prove directly from the definition that $\mathcal{F}_{{\rm GZ},u}^{-1}$
satisfies the 2-cocycle condition. In Section~\ref{ssec:star}, we compute the corresponding star product and in Section~\ref{ssec:twcop} the twisted coproduct $\Delta p_\mu$. In Section~\ref{ssec:nccoor} we introduce noncommutative coordinates and their realizations. Section~\ref{sec:cob} is dedicated to the family $\mathcal{F}_{R,u}$ of Jordanian twists obtained from a simple Jordanian twist $\mathcal{F}_0$~(\ref{eq:F0F1}) via twisting by a 1-cochain. We start the section by introducing $\mathcal{F}_{R,u}$ as a product of three exponential factors. Then we compute the corresponding deformed Hopf algebra in Section~\ref{ssec:Hopf}, introduce the corresponding noncommutative coordinates and realizations in Section~\ref{ssec:nccoorcob} and compute the star products in Section~\ref{ssec:starcob}. In Section~\ref{sec:equality}, we present two different proofs both showing that $\mathcal{F}_{{\rm GZ},u}$ equals $\mathcal{F}_{R,u}$. The first proof in Section~\ref{ssec:diff} is by showing that they solve the same Cauchy problem (an ordinary differential equation with initial condition). The second proof in Section~\ref{ssec:proofstar} uses a comparison among the star products. The final Section~\ref{sec:concl} is the conclusion. Appendix~\ref{appendixA} is added presenting a proof of an identity used in the proof in Section~\ref{ssec:coc} of the 2-cocycle condition for $\mathcal{F}^{-1}_{{\rm GZ},u}$.

\section{Generalization of the Giaquinto--Zhang twist}
\label{sec:GZ}
We define the generalized Jordanian twist $\mathcal{F}_{{\rm GZ},u}^{-1}$
via an explicit expansion,
\begin{gather}\label{eq:FGZu}
\mathcal{F}_{{\rm GZ},u}^{-1}=\sum_{k,l=0}^\infty \left(\frac{1}{\kappa}\right)^{k+l} \left( (u-1)P\right)^k \binom{D}{l} \otimes (uP)^l \binom{D}{k}.
\end{gather}
The twist $\mathcal{F}_{{\rm GZ},u}^{-1}$ interpolates between $\mathcal{F}_0^{-1}$ and $\mathcal{F}_1^{-1}$. For $u\rightarrow 0$ one can easily see~\cite{BP} that~(\ref{eq:FGZu}) reduces to
\begin{gather*}
\mathcal{F}_0^{-1}=\sum_{m=0}^\infty \left(\frac{-1}{\kappa}\right)^m P^m \otimes \binom D m = {\rm e}^{\ln ( 1-\frac{1}{\kappa}P )\otimes D}.
\end{gather*}
For $u=1$
\begin{gather*}
\mathcal{F}_1^{-1}= \sum_{m=0}^\infty \left( \frac{1}{\kappa}\right)^m \binom D m \otimes P^m= {\rm e}^{D\otimes \ln ( 1+\frac{1}{\kappa}P )}.
\end{gather*}
For $u=\frac{1}{2}$ this reduces to the twist introduced in~\cite{GZ}, where
\begin{gather*}
t=\frac{1}{2\kappa}, \qquad E=P,\qquad \text{and} \qquad H=-D.
\end{gather*}

\subsection{2-cocycle condition}\label{ssec:coc}
\begin{theorem} For arbitrary $u$, twists $\mathcal{F}_{{\rm GZ},u}^{-1}$ satisfy the $2$-cocycle condition given by
\begin{gather}\label{eq:cocycFm}
\big( (\Delta_0\otimes 1) \mathcal{F}^{-1}_{{\rm GZ},u}\big) \big(\mathcal{F}^{-1}_{{\rm GZ},u}\otimes 1\big)
= \big( (1\otimes\Delta_0) \mathcal{F}^{-1}_{{\rm GZ},u}\big) \big(1\otimes\mathcal{F}^{-1}_{{\rm GZ},u}\big).
\end{gather}
\end{theorem}

\begin{proof} If we write
\begin{gather*}
f_n := \sum_{k + l = n}\left( (u-1)P\right)^k \binom{D}{l} \otimes (uP)^l \binom{D}{k}
\end{gather*}
then
\begin{gather*}\mathcal{F}^{-1}_{{\rm GZ},u} = \sum_{n = 0}^\infty\left(\frac{1}{\kappa}\right)^n f_n\end{gather*}
with $f_n$ not depending on $\kappa$ and $f_0 = 1\otimes 1$.
In terms of $f_i$, the 2-cocycle condition becomes a~sequence of equations for all $n$,
\begin{gather*}
\sum_{i = 0}^n \left( (\Delta_0\otimes 1) f_i\right) (f_{n-i}\otimes 1)= \sum_{i = 0}^n \left( (1\otimes\Delta_0)f_i\right) (1\otimes f_{n-i}).
\end{gather*}

In the first order in $1/\kappa$,
\begin{gather*}
f_1 = (u - 1)P\otimes D + u D\otimes P,\\
(\Delta_0\otimes 1) f_1 + f_1\otimes 1 =
(1\otimes\Delta_0) f_1 + 1\otimes f_1,
\end{gather*}
and in the second order,
\begin{gather*}
f_2 = (u - 1)^2 P^2 \otimes \binom{D}{2} +(u-1) u P D\otimes P D + u^2\binom{D}{2}\otimes P^2,\\
(\Delta_0\otimes 1) f_2 + ((\Delta_0\otimes 1) f_1 )(f_1\otimes 1)+ f_2\otimes 1=(1\otimes\Delta_0) f_2 + ( (1\otimes\Delta_0)f_1 )(1\otimes f_1)+ 1\otimes f_2.
\end{gather*}
For general order $n$, it should hold that
\begin{gather*}
\sum_{\substack{k_1,k_2, l_1,l_2 = 0
 \\ k_1+k_2 = k,\, l_1+l_2 = l,\, k+l = n}}^\infty
\left[\Delta_0\left(P^{k_1}\binom{D}{l_1}\right)\otimes P^{l_1}\binom{D}{k_1}\right] \left[ P^{k_2}\binom{D}{l_2}\otimes P^{l_2}\binom{D}{k_2}\otimes 1\right]\\
\qquad{} = \sum_{\substack{k_1,k_2, l_1,l_2 = 0
 \\ k_1+k_2 = k,\, l_1+l_2 = l,\, k+l = n}}^\infty \left[ P^{k_1}\binom{D}{l_1}\otimes\Delta_0\left(P^{l_1}\binom{D}{k_1}\right)\right]\left[
 1\otimes P^{k_2}\binom{D}{l_2}\otimes P^{l_2}\binom{D}{k_2}\right].
\end{gather*}
This can be rewritten as
\begin{gather*}
 \sum_{k_1 = 0}^k\sum_{l_1 = 0}^l
 \left[\Delta_0\left(P^{k_1}\binom{D}{l_1}\right)\otimes P^{l_1}\binom{D}{k_1}\right]
 \left[P^{k-k_1}\binom{D}{l-l_1}\otimes P^{l-l_1}\binom{D}{k-k_1}\otimes 1\right] \\
 \qquad{} =\sum_{k_1 = 0}^k\sum_{l_1 = 0}^l
\left[P^{k_1}\binom{D}{l_1}\otimes\Delta_0\left(P^{l_1}\binom{D}{k_1}\right)\right] \left[1\otimes P^{k-k_1}\binom{D}{l-l_1}\otimes P^{l-l_1}\binom{D}{k-k_1}\right],\\
 \sum_{k_1 = 0}^k\sum_{l_1 = 0}^l \left[\Delta_0\big(P^{k_1}\big)\big(P^{k-k_1}\otimes P^{l_2-l_1}\big)
 \binom{(D-k+k_1)\otimes 1 + 1\otimes(D-l+l_1)}{l_1}\otimes P^{l_1}\binom{D}{k_1}\right]\\
 \qquad\quad {}\times \left[\binom{D}{l-l_1}\otimes\binom{D}{k-k_1}\otimes 1\right] \\
 \qquad{} = \sum_{k_1 = 0}^k\sum_{l_1 = 0}^l \big[P^{k_1}\otimes\big(\Delta_0\big(P^{l_1}\big)
 \big(P^{k-k_1}\otimes P^{l-l_1}\big)\big)\big]\\
 \qquad\quad{}\times \left[\binom{D}{l_1} \otimes\binom{(D-k+k_1)\otimes 1 + 1\otimes (D-l+l_1)}{k_1}
 \left(\binom{D}{l-l_1}\otimes\binom{D}{k-k_1}\right)\right].
\end{gather*}

Let us compare the terms of type $P^A\otimes P^B\otimes P^C$ with $A + B + C = k + l = n$ on both sides. We see only the terms with $C = l_1$ on the left-hand side and only the terms with $A = k_1$ on the right-hand side. We also need to take into account $\Delta_0\big(P^{k_1}\big) = \sum\limits_{a = 0}^{k_1} \binom{k_1}{a} P^{k_1-a}\otimes P^a$ on the left-hand side and $\Delta_0\big(P^{l_1}\big) = \sum\limits_{b = 0}^{l_1}\binom{l_1}{b} P^b\otimes P^{l_1-b}$ on the right-hand side to obtain
\begin{gather*}
\sum_{\substack{k_1 = 0\\ k - a = A}}^k\sum_{a = 0}^{k_1}
\binom{k_1}{a}\binom{(D-k+k_1)\otimes 1 + 1\otimes(D-l+C)}{C}\left(\binom{D}{l-C}\otimes\binom{D}{k-k_1}\right)\otimes\binom{D}{k_1}
\\=\sum_{\substack{l_1 = 0\\ l - b = C}}^l\sum_{b = 0}^{l_1}
\binom{l_1}{b}\binom{D}{l_1}\otimes\binom{(D-k+A)\otimes 1 + 1\otimes(D-l+l_1)}{A}
\left(\binom{D}{l-l_1}\otimes\binom{D}{k-A}\right),\\
\sum_{\substack{k - a = A\\ k_1 = 0}}^k \binom{k_1}{k-A}\binom{(D-k+k_1)\otimes 1 + 1\otimes (D-l+C)}{C}
\left(\binom{D}{l-C}\otimes\binom{D}{k-k_1}\right)\otimes\binom{D}{k_1}\\
= \sum_{\substack{l- b = C\\ l_1 = 0}}^l \binom{l_1}{l-C}\binom{D}{l_1} \otimes
\binom{(D-k+A)\otimes 1 + 1 \otimes (D-l+l_1)}{A}\left(\binom{D}{l-l_1}\otimes
 \binom{D}{k-A}\right)
\end{gather*}
for every $k,l\in\mathbf{N}_0$, and all $A\leq k$, $C\leq l$.

In terms of the new variables
\begin{gather*}
x = D\otimes 1 \otimes 1,\qquad y = 1 \otimes D \otimes 1,\qquad z = 1\otimes 1 \otimes D,
\end{gather*}
and taking into account that $k+l=n$, we reduce the 2-cocycle condition to the identity
\begin{gather}
\binom{x}{l-C} \sum_{k_1 = k-A}^k \binom{k_1}{k-A}\binom{x+y-k-l+k_1+C}{C}\binom{y}{k-k_1}\binom{z}{k_1}\nonumber \\
\qquad{}=\binom{z}{k-A}\sum_{l_1 = l-C}^l \binom{l_1}{l-C}
\binom{y+z-k-l+l_1+A}{A}\binom{x}{l_1}\binom{y}{l-l_1}\label{eq:mainident}
\end{gather}
for all $A \leq k$ and $C \leq l$. This is restated as~(\ref{eq:bigident}), and then proved, in Appendix~\ref{appendixA}.
\end{proof}

For $C = 0$ the identity~(\ref{eq:mainident}) reduces to
\begin{gather*}
\sum_{k_1 = k - A}^k \binom{k_1}{k-A} \binom{y}{k-k_1} \binom{z}{k_1} = \binom{z}{k-A} \binom{y+z-k+A}{A}.
\end{gather*}

\subsection{Star product}\label{ssec:star}

We now introduce an action $\triangleright$ of $P$ and $D$ on the space of formal power series in variables $x_\mu$, where $\mu = 0,1,\ldots,n$, by formulas
\begin{gather*}
 (P\triangleright f)(x) = -{\rm i} v_\mu\frac{\partial f(x)}{\partial x_\mu},\qquad (D\triangleright f)(x) = x_\mu \frac{\partial f(x)}{\partial x_\mu},
\end{gather*}
where the constants $v_\mu$ are such that $v^2\in\{-1,0,1\}$ and the Einstein summation rule is understood. We also denote $x = (x_\mu)$ and $\partial_\mu = \frac{\partial}{\partial x_\mu}$.

A star product $*$ is then defined as
\begin{gather*}
f * g=m \big( \mathcal{F}_{{\rm GZ},u}^{-1} (\triangleright \otimes \triangleright)(f\otimes g) \big)
\end{gather*}
for all formal power series $f$, $g$ in $x_\mu$~\cite{mercati}. In particular, for $f={\rm e}^{{\rm i}kx}$ and $g={\rm e}^{{\rm i}qx}$,
\begin{gather}\label{eq:expstarexp}
{\rm e}^{{\rm i}kx} * {\rm e}^{{\rm i}qx} =m\big( \mathcal{F}_{{\rm GZ},u}^{-1} (\triangleright \otimes \triangleright) \big({\rm e}^{{\rm i}kx} \otimes {\rm e}^{{\rm i}qx} \big) \big) =: {\rm e}^{\AA(u;k,q,x)},
\end{gather}
where
\begin{gather*}
kx = k_\alpha x_\alpha\qquad \mathrm{and}\qquad qx=q_\alpha x_\alpha
\end{gather*}
are elements of the Minkowski space-time algebra, the function $\mathcal{A}$ is implicitly defined by~(\ref{eq:expstarexp}) and $m$ denotes the multiplication map on usual functions. Using the actions of $P$ and $D$ on ${\rm e}^{{\rm i}kx}$, it follows that
\begin{gather*}
P\triangleright {\rm e}^{{\rm i}kx} = (v\cdot k) {\rm e}^{{\rm i}kx},\qquad
P\triangleright {\rm e}^{{\rm i}qx} = (v\cdot q) {\rm e}^{{\rm i}qx},
\end{gather*}
where $(v\cdot k) = v_\alpha k_\alpha$, $(v\cdot q) = v_\alpha q_\alpha$.
For $j<l$,
\begin{gather*}
\left.\left(P^j \binom D l \triangleright {\rm e}^{{\rm i}kx}\right)\right|_{x=0} = 0, \\
\left.\left(P^n \binom D n \triangleright {\rm e}^{{\rm i}kx}\right)\right|_{x=0} = (v\cdot k)^n, \\
\left.\left(P^{n+1} \binom D n \triangleright {\rm e}^{{\rm i}kx}\right)\right|_{x=0} = (n+1)(v\cdot k)^{n+1}.
\end{gather*}
The following identities hold
\begin{gather*}
\binom D n \triangleright {\rm e}^{{\rm i}kx} = \frac{({\rm i}kx)^n}{n!}{\rm e}^{{\rm i}kx}, \qquad n\in\mathbf{N}_0,\\
\left.\left(P^j \binom D l \triangleright {\rm e}^{{\rm i}kx}\right)\right|_{x=0} =
\left.\left(P^j \triangleright \frac{({\rm i}kx)^l}{l!}{\rm e}^{{\rm i}kx}\right)\right|_{x=0}
= \binom j l (v_\alpha k_\alpha)^j.
\end{gather*}
Then we have
\begin{gather}
 \big({\rm e}^{{\rm i}kx}* {\rm e}^{{\rm i}qx}\big)\big\vert_{x=0} = \sum_{n=0}^\infty \left( \frac{u(u-1)}{\kappa^2}\right)^n (v\cdot k)^n (v\cdot q)^n\nonumber\\
 \hphantom{\big({\rm e}^{{\rm i}kx}* {\rm e}^{{\rm i}qx}\big)\big\vert_{x=0}}{} =\frac{1}{1-\frac{u(u-1)}{\kappa^2}(v\cdot k)(v\cdot q)} = {\rm e}^{\AA(u;k,q,x)}\left.\right\vert_{x=0}. \label{eq:expstarexpx0}
\end{gather}
We now calculate the partial derivatives of the star product,
\begin{gather*}
 \partial_\mu \big({\rm e}^{{\rm i}kx} * {\rm e}^{{\rm i}qx}\big)\big\vert_{x=0}=
 \sum_{j,l=0}^\infty \left(\frac{1}{\kappa}\right)^{j+l}
 \left\{ {\rm i}k_\mu ((u-1)P)^j\triangleright \left[\frac{(v\cdot k)^l}{l!}+\frac{(v\cdot k)^{l-1}}{(l-1)!}\right]{\rm e}^{{\rm i}kx}\right\}\\
 {}\times \left\{ (uP)^l \triangleright \frac{(v\cdot q)^j}{j!}{\rm e}^{{\rm i}qx}\right\}\bigg|_{x=0} \\
 + \sum_{j,l = 0}^\infty\left( \frac{1}{\kappa}\right)^{j+l} \left.\left\{
((u-1)P)^j\triangleright \frac{(v\cdot k)^l}{ l!}{\rm e}^{{\rm i}kx}
 \right\}
 \left\{
 {\rm i}q_\mu (v P)^l\triangleright \left[\frac{(v\cdot q)^j}{j!}+\frac{(v\cdot q)^{j-1}}{(j-1)!}{\rm e}^{{\rm i}qx}\right]
 \right\}\right|_{x=0}\\
 = {\rm i}(k_\mu+q_\mu) \sum_{n=0}^\infty \left(\frac{u(u-1)}{\kappa^2}\right)^n (v\cdot k)^n(v\cdot q)^n + {\rm i}(k_\mu+q_\mu)\sum_{n=0}^\infty \left(\frac{u(u-1)}{\kappa^2}\right)^n n (v\cdot k)^n\\
+ (v \cdot q)^n +{\rm i}\left[k_\mu \frac{u}{\kappa}(v\cdot q)+q_\mu \frac{(u-1)}{\kappa}(v\cdot k)\right]\sum_{n=0}^\infty \left(\frac{u(u-1)}{\kappa^2}\right)^n (n+1) (v\cdot k)^n(v\cdot q)^n\\
= {\rm i} \mathcal{D}_\mu(k,q) \frac{1}{1-\frac{u(u-1)}{\kappa^2}(v\cdot k)(v\cdot q)}
\end{gather*}
and
\begin{gather*}
 \partial_\mu \big({\rm e}^{{\rm i}kx} * {\rm e}^{{\rm i}qx}\big)\big\vert_{x=0}= {\rm i}\left( k_\mu \left( 1+\frac{u}{\kappa}(v\cdot q)\right) +q_\mu \left( 1+\frac{u-1}{\kappa}(v\cdot k)\right) \right) \\
 \hphantom{\partial_\mu \big({\rm e}^{{\rm i}kx} * {\rm e}^{{\rm i}qx}\big)\big\vert_{x=0}=}{}\times \sum_{n=0}^\infty \left(\frac{u(u-1)}{\kappa^2}\right)^n (n+1)(v\cdot k)^n(v\cdot q)^n \\
\hphantom{\partial_\mu \big({\rm e}^{{\rm i}kx} * {\rm e}^{{\rm i}qx}\big)\big\vert_{x=0}}{} = {\rm i}\mathcal{D}_\mu (k,q) \frac{1}{1-\frac{u(u-1)}{\kappa^2}(v\cdot k)(v\cdot q)}= {\rm i}\mathcal{D}_\mu (k,q)\big({\rm e}^{{\rm i}kx} * {\rm e}^{{\rm i}qx}\big)\big\vert_{x=0}.
\end{gather*}
Note that
\begin{gather*}
{\rm i} \mathcal{D}_\mu (k,q) = \left.\left(\frac{\partial\AA(u; k,q,x) }{\partial x_\mu} \right)\right|_{x=0}.
\end{gather*}
It follows that
\begin{gather}\label{eq:Dmu}
\mathcal{D}_\mu (k,q)= \frac{k_\mu \big(1+\frac{u}{\kappa}(v\cdot q)\big) +q_\mu \big( 1+\frac{u-1}{\kappa} (v\cdot k)\big)}{1-\frac{u(u-1)}{\kappa^2}(v\cdot k) (v\cdot q)}.
\end{gather}

\subsection[Twisted coproduct $\Delta(p_\mu)$]{Twisted coproduct $\boldsymbol{\Delta(p_\mu)}$}\label{ssec:twcop}

Let now $p_\mu = -{\rm i}\partial_\mu$ be the momentum operator. Let us define $\Delta p_\mu$ by
\begin{gather}\label{eq:Deltapmu}
\Delta p_\mu =\mathcal{D}_\mu (p\otimes 1,1\otimes p)= \frac{p_\mu \otimes \big(1+\frac{u}{\kappa}P\big) +\big(1+\frac{u-1}{\kappa}P\big) \otimes p_\mu }{1\otimes 1-\frac{u(u-1)}{\kappa^2}P\otimes P}, \qquad P=v_\alpha p_\alpha.
\end{gather}
We want to show that $\Delta p_\mu $ is the deformed coproduct with respect to the twist $\mathcal{F}_{{\rm GZ},u}$,
\begin{gather}\label{eq:FDelta0pmuF}
\Delta p_\mu =\mathcal{F}_{{\rm GZ},u}\Delta_0 p_\mu \mathcal{F}_{{\rm GZ},u}^{-1},
\end{gather}
where
\begin{gather}\label{eq:Delta0pmu}
\Delta_0 p_\mu =p_\mu \otimes 1+1\otimes p_\mu.
\end{gather}
Using~(\ref{eq:Delta0pmu}) and~(\ref{eq:Deltapmu}), we may rewrite~(\ref{eq:FDelta0pmuF}) as
\begin{gather*}
 \mathcal{F}_{{\rm GZ},u}^{-1} \frac{p_\mu \otimes \big(1+\frac{u}{\kappa}P\big) +\big(1+\frac{u-1}{\kappa}P\big) \otimes p_\mu }{1\otimes 1-\frac{u(u-1)}{\kappa^2}P\otimes P} =
 (p_\mu\otimes 1+1\otimes p_\mu) \mathcal{F}_{{\rm GZ},u}^{-1}
\end{gather*}
and, after multiplying from the right by the denominator $1\otimes 1-\frac{u(u-1)}{\kappa^2}P\otimes P$, as
 \begin{gather}
\mathcal{F}_{{\rm GZ},u}^{-1} \left(p_\mu \otimes \left(1+\frac{u}{\kappa}P\right) +\left(1+\frac{u-1}{\kappa}P\right) \otimes p_\mu\right)\nonumber\\
\qquad{} = (p_\mu\otimes 1+1\otimes p_\mu) \mathcal{F}_{{\rm GZ},u}^{-1}\left(1\otimes 1-\frac{u(u-1)}{\kappa^2}P\otimes P\right).\label{eq:twoterms}
 \end{gather}
We shall show the equality in~(\ref{eq:twoterms}) by splitting it into a sum of two equalities, (\ref{eq:Fpotimes1}) and~(\ref{eq:F1otimesp}), which are then separately proved. Descriptively, (\ref{eq:Fpotimes1}) involves all those summands in expanded~(\ref{eq:twoterms}) where, in one of the factors, $p_\mu$ is at the left side from the tensor product,
\begin{gather}\label{eq:Fpotimes1}
\mathcal{F}_{{\rm GZ},u}^{-1} \left( p_\mu \otimes \left( 1+\frac{u}{\kappa}P\right) \right) =(p_\mu\otimes 1) \mathcal{F}_{{\rm GZ},u}^{-1} \left( 1\otimes 1- \frac{u(u-1)}{\kappa^2} P\otimes P\right).
\end{gather}
To prove this equality, we first observe that by induction the equality $[P,D]=P$ implies the commutation relation	
\begin{gather*}
p_\mu \binom D k =\binom{D+1}{k} p_\mu.
\end{gather*}
Hence
\begin{gather*}
P \binom D k =\binom{D+1}{k} P,
\end{gather*}
i.e.,
\begin{gather*}
\binom D k P=P \binom{D-1}{k}.
\end{gather*}
We calculate the left-hand side of~(\ref{eq:Fpotimes1}) as
\begin{gather*}
 \mathcal{F}_{{\rm GZ},u}^{-1}\left(p_\mu \otimes \left( 1+\frac{u}{\kappa}P\right)\right) \\
 \qquad {} =\sum_{k,l=0}^\infty \left( \left(\frac{u-1}{\kappa}P\right)^k \binom D l \otimes \left( \frac{uP}{\kappa}\right)^l \binom D k \right) \left( p_\mu \otimes 1+\frac{u}{\kappa} p_\mu \otimes P\right) \\
 \qquad{} =\sum_{k,l=0}^\infty \left( \left( \frac{(u-1)P}{\kappa}\right)^k \binom D l \otimes \left(\frac{uP}{\kappa}\right)^l \binom D k \right.\\
\left.\qquad\quad{} +\left(\frac{(u-1)P}{\kappa}\right)^k \binom D l \otimes \left( \frac{uP}{\kappa}\right)^{l+1} \binom{D-1}{k} \right) (p_\mu \otimes 1)
\end{gather*}
and the right-hand side of~(\ref{eq:Fpotimes1}) as
\begin{gather*}
 (p_\mu \otimes 1) \mathcal{F}_{{\rm GZ},u}^{-1} \left(1\otimes 1 -\frac{u(u-1)}{\kappa^2}P\otimes P\right) \nonumber \\
\qquad{} = \sum_{k,l=0}^\infty \left( \left( \frac{(u-1)P}{\kappa} \right)^k \binom{D+1}{l} \otimes \left(\frac{uP}{\kappa} \right)^l \binom D k \right.\\
\left. \qquad\quad{} -\left( \frac{(u-1)P}{\kappa}\right)^{k+1} \binom D l \otimes \left( \frac{uP}{\kappa}\right)^{l+1} \binom{D-1}{k} \right) (p_\mu \otimes 1).
\end{gather*}
Comparing the terms of type $P^k\otimes P^l$ for all $k$ and $l$, we find
\begin{gather*}
 \left( \frac{(u-1)P}{\kappa}\right)^k \binom D l \otimes \left( \frac{uP}{\kappa}\right)^l \binom D k + \left( \frac{(u-1)P}{\kappa} \right)^k l \binom{D}{l-1} \otimes \left(\frac{uP}{\kappa}\right)^l \binom D k \frac{D-k}{D} \\
\qquad{} =\left( \frac{(u-1)P}{\kappa} \right)^k (D+1)\binom{D}{l-1} \otimes \left(\frac{uP}{\kappa}\right)^l \binom D k\\
\qquad\quad{} - \left( \frac{(u-1)P}{\kappa} \right)^k l \binom{D}{l-1} \otimes \left( \frac{uP}{\kappa}\right)^l \frac{k}{D} \binom D k
\end{gather*}
and
\begin{gather*}
 \left( \frac{(u-1)P}{\kappa}\right)^k \binom D l \otimes \left( \frac{uP}{\kappa}\right)^l \binom D k\nonumber \\
\qquad{}= \left( \frac{(u-1)P}{\kappa} \right)^k l \binom{D}{l-1} \otimes \left( \frac{uP}{\kappa}\right)^l \binom D k \\
\qquad\quad{} +\left( \frac{(u-1)P}{\kappa}\right)^k \!(D+1) \binom{D}{l-1} \otimes \left( \frac{uP}{\kappa}\right)^l \binom D k \nonumber \\
\qquad{} = \left( \frac{u-1)P}{\kappa}\right)^k \binom{D}{l-1} (D-l+1) \otimes \left( \frac{uP}{\kappa}\right)^l \binom D k\nonumber \\
\qquad{} = \left( \frac{(u-1)P}{\kappa}\right)^k \binom D l\otimes \left( \frac{uP}{\kappa} \right)^l \binom D k,
\end{gather*}
which proves~(\ref{eq:Fpotimes1}).

Analogously, we prove the equality of the remaining summands in~(\ref{eq:twoterms}),
\begin{gather}\label{eq:F1otimesp}
\mathcal{F}_{{\rm GZ},u}^{-1} \left( \left(1+(u-1)\frac{1}{\kappa}P\right) \otimes p_\mu \right) = (1\otimes p_\mu) \mathcal{F}_{{\rm GZ},u} ^{-1} \left(1\otimes 1-\frac{u(u-1)}{\kappa^2} P\otimes P\right).
\end{gather}
Now~(\ref{eq:Fpotimes1}) and~(\ref{eq:F1otimesp}) add to~(\ref{eq:twoterms}). Hence this proves~(\ref{eq:FDelta0pmuF}), that is
\begin{gather*}
\Delta p_\mu =\mathcal{F}_{{\rm GZ},u} \Delta_0 (p_\mu) \mathcal{F}_{{\rm GZ},u}^{-1}.
\end{gather*}
The coproduct $\Delta p_\mu$ satisfies the coassociativity condition
\begin{gather*}
(\Delta \otimes 1)\Delta p_\mu =(1\otimes \Delta )\Delta p_\mu.
\end{gather*}
Equation~(\ref{eq:FDelta0pmuF}) can be rewritten as
\begin{gather}\label{eq:D0pFeqFDp}
 \Delta_0 p_\mu \mathcal{F}_{{\rm GZ},u}^{-1}=\mathcal{F}_{{\rm GZ},u}^{-1} \Delta p_\mu.
\end{gather}

This enables us to obtain explicit formulas for the derivatives of the star product and for the star product from Section~\ref{ssec:star}. Namely, for the partial derivatives of the star product, we compute
\begin{gather}
 \partial_\mu \big({\rm e}^{{\rm i}kx} * {\rm e}^{{\rm i}qx}\big) \overset{(\ref{eq:expstarexp})}{=}
 \partial_\mu m \mathcal{F}_{{\rm GZ},u}^{-1} \big({\rm e}^{{\rm i}kx}\otimes {\rm e}^{{\rm i}qx}\big)
= m (\partial_\mu\otimes 1+1\otimes \partial_\mu) \mathcal{F}_{{\rm GZ},u}^{-1}
 \big({\rm e}^{{\rm i}kx}\otimes {\rm e}^{{\rm i}qx}\big)\nonumber \\
\hphantom{\partial_\mu \big({\rm e}^{{\rm i}kx} * {\rm e}^{{\rm i}qx}\big)}{}
 = m {\rm i}\Delta_0(p_\mu)\mathcal{F}_{{\rm GZ},u}^{-1} \big({\rm e}^{{\rm i}kx}\otimes {\rm e}^{{\rm i}qx}\big)
 \overset{(\ref{eq:D0pFeqFDp})}{=} {\rm i} m\mathcal{F}_{{\rm GZ},u}^{-1}\Delta(p_\mu) \big({\rm e}^{{\rm i}kx}\otimes {\rm e}^{{\rm i}qx}\big)\nonumber\\
\hphantom{\partial_\mu \big({\rm e}^{{\rm i}kx} * {\rm e}^{{\rm i}qx}\big)}{}
 \overset{(\ref{eq:Deltapmu})}{=} m {\rm i}\mathcal{F}_{{\rm GZ},u}^{-1}\mathcal{D}_\mu(p\otimes 1,1\otimes p)\big({\rm e}^{{\rm i}kx}\otimes {\rm e}^{{\rm i}qx}\big)
 = {\rm i} m\mathcal{F}_{{\rm GZ},u}^{-1}\mathcal{D}_\mu(k,q)\big({\rm e}^{{\rm i}kx}\otimes {\rm e}^{{\rm i}qx}\big)\nonumber\\
 \hphantom{\partial_\mu \big({\rm e}^{{\rm i}kx} * {\rm e}^{{\rm i}qx}\big)}{}
 = {\rm i}\mathcal{D}_\mu (k,q) m \mathcal{F}_{{\rm GZ},u}^{-1}\big({\rm e}^{{\rm i}kx} \otimes {\rm e}^{{\rm i}qx}\big)
 \overset{(\ref{eq:expstarexp})}{=} i\mathcal{D}_\mu (k,q) \big({\rm e}^{{\rm i}kx}*{\rm e}^{{\rm i}qx}\big),\label{eq:partexpstarexp}
\end{gather}
where $m$ denotes the multiplication map for usual functions. Knowing the partial derivatives~(\ref{eq:partexpstarexp}) and the initial value~(\ref{eq:expstarexpx0}) of the star product at $x = 0$, we finally obtain
\begin{gather}\label{eq:eikxstar}
{\rm e}^{{\rm i}kx} * {\rm e}^{{\rm i}qx} = {\rm e}^{{\rm i}\mathcal{D}_\mu (k,q)x^\mu } \frac{1}{1-\frac{u(u-1)}{\kappa^2}(v\cdot k) (v\cdot q)},
\end{gather}
where $\mathcal{D}_\mu (k,q)$ is given in~(\ref{eq:Dmu}). This star product is associative in agreement with the fact that twists $\mathcal{F}^{-1}_{{\rm GZ},u}$ satisfy the 2-cocycle condition~(\ref{eq:cocycFm}).

\subsection{Noncommutative coordinates and realizations} \label{ssec:nccoor}

Here we introduce noncommutative coordinates $\hat{x}_\mu$, the commutation relations among them and their realizations. We use realizations of elements of noncommutative algebras via a Heisenberg algebra with generators $x_\mu$, $p_\nu$, $[x_\mu,x_\nu] = 0$, $[p_\mu,p_\nu] = 0$, $[x_\mu,p_\nu] = -{\rm i}\delta_{\mu,\nu}$. The following expression defines noncommutative coordinates $\hat{x}_\mu$~\cite{mercati},
\begin{gather*}
\hat{x}_\mu = m\big( \mathcal{F}_{{\rm GZ},u}^{-1} (\triangleright \otimes 1)(x_\mu \otimes 1)\big)
 = x_\mu \left(1+\frac{u}{\kappa}P\right) +\frac{{\rm i}}{\kappa} v_\mu (1-u)\left( 1+\frac{u}{\kappa}P\right) D \\
\hphantom{\hat{x}_\mu}{} = \left( x_\mu +(1-u) \frac{{\rm i}}{\kappa} v_\mu D\right) \left( 1+\frac{u}{\kappa}P\right) +\frac{u(1-u)}{\kappa^2} {\rm i}v_\mu P.
\end{gather*}
Noncommutative coordinates $\hat{x}_\mu$ satisfy a $\kappa$-deformed Heisenberg algebra that corresponds to the $\kappa$-Minkowski space~\cite{govindarajan,EPJC2015,pikuticEPJC2017,LukRuegg,LukTol,mkj,mmss2,stojic2}
\begin{gather*}
[\hat{x}_\mu,\hat{x}_\nu] = \frac{{\rm i}}{\kappa} (v_\mu \hat{x}_\nu -v_\nu \hat{x}_\mu),\\
 [ p_\mu, \hat{x}_\nu ] = \left(-{\rm i}\delta_{\mu,\nu} +\frac{{\rm i}}{\kappa} v_\nu (1-u)p_\mu \right) \left( 1+\frac{u}{\kappa} P\right).
\end{gather*}
In the case $u=0$,
\begin{gather*}
\hat{x}_\mu =x_\mu + \frac{{\rm i}}{\kappa} v_\mu D.
\end{gather*}
In the case $u=1$,
\begin{gather*}
\hat{x}_\mu = x_\mu \left(1+\frac{u}{\kappa}P\right).
\end{gather*}
Using this realization of $\hat{x}_\mu$ and the method from~\cite{mercati} we obtain the same star product~(\ref{eq:eikxstar}).

\section[Interpolation between Jordanian twists induced by a 1-cochain]{Interpolation between Jordanian twists\\ induced by a 1-cochain}\label{sec:cob}

Another construction for a generalized Jordanian twist is possible~\cite{cobtw}. This twist, here denoted~$\mathcal{F}_{R,u}$, has been introduced as a product of three exponential factors,
\begin{gather}\label{eq:FRuexp}
\mathcal{F}_{R,u}=\exp \left(\frac{u}{\kappa}(PD\otimes 1+1 \otimes PD)\right) \exp \left(\!{-}\ln\left(1-\frac{1}{\kappa}P\right) \otimes D \right) \exp \left( \Delta_0\left(\!{-}\frac{u}{\kappa}PD\right)\right),\!\!\!\!
\end{gather}
where $u$ is a real parameter, $u\in\mathbb{R}$. The symbol $R$ in the subscript refers to the position of the dilatation generator in the formula, namely it is on the right with respect to~$P$. The classical $r$-matrix corresponding to twists $F_{{\rm GZ},u}$~(\ref{eq:FGZu}) and $F_{R,u}$~(\ref{eq:FRuexp}) does not depend on the parameter~$u$, namely
\begin{gather*}
r = \frac{1}{\kappa}(D\otimes P - P\otimes D).
\end{gather*}

The above form~(\ref{eq:FRuexp}) of the family of twists $\mathcal{F}_{R,u}$ is obtained from a simple Jordanian twist~$\mathcal{F}_0$, using a transformation by a 1-cochain. Namely, according to Drinfeld~\cite{drinfeld,majid}, if~$\mathcal{F}$ is any normalized Drinfeld twist and $\omega_R$ is any element in the Hopf algebra satisfying the normalization $\epsilon(\omega_R) = 1$, then the formula $\mathcal{F}_\omega := \big(\omega^{-1}\otimes\omega^{-1}\big) \mathcal{F}\Delta(\omega)$ defines a normalized Drinfeld twist again (that is, the 2-cocycle and counitality conditions are satisfied again). In particular, if $\mathcal{F} = 1\otimes 1$ we get a 2-coboundary twist $\big(\omega^{-1}\otimes\omega^{-1}\big)\Delta(\omega)$. If the two twists, $\mathcal{F}$ and $\mathcal{F}_\omega$, transform one into another by a 1-cochain, we say that they are cohomologous in the sense of nonabelian cohomology~\cite{majid}. In this case, twisted Hopf algebras $H_{\mathcal{F}}$ and $H_{\mathcal{F}_\omega}$ are isomorphic~\cite{majid} and, for each $H$-module algebra $M$, the corresponding twistings $M_{\mathcal{F}}$ and $M_{\mathcal{F}_\omega}$ are also mutually isomorphic as algebras. If~$\omega$ is group like, $\mathcal{F}_\omega$ is evidently obtained from~$\mathcal{F}$ by an inner automorphism. Regarding that cohomologous twists give isomorphic mathematical objects, one sometimes thinks of these twists as gauge equivalent.

If $\mathcal{F} = \mathcal{F}_0$ is a simple Jordanian twist, and $\omega = \omega_R = \exp \big( {-}\frac{u}{\kappa}PD \big)$, we obtain the twist $\mathcal{F}_{R,u}=\mathcal{F}_{\omega_R}=\big(\omega^{-1}_R\otimes\omega^{-1}_R\big)\mathcal{F}_0\Delta(\omega_R)$, see~\cite{cobtw}. This also shows that, for any~$u$, twist~$\mathcal{F}_{R,u}$ satisfies the 2-cocycle and normalization conditions. Regarding that $u$ appeared by gauge transforming~$\mathcal{F}_0$, we can view $u$ as a gauge parameter (the reader should not confuse $u$ with a spectral parameter involved in some other Jordanian deformations). For $u=0$, twist $\mathcal{F}_{R,u}$ simplifies to~$\mathcal{F}_0$ and for $u=1$ to $\mathcal{F}_1$.

\subsection{Hopf algebra}\label{ssec:Hopf}

The coalgebra sector of the Hopf algebra $\mathcal{H}^{\mathcal{F}_{R,u}}$ for the deformation with $\mathcal{F}_{R,u}$ is given by the formulas
\begin{gather*}
\Delta^{\mathcal{F}_{R,u}} p_\mu = \frac{p_\mu \otimes \big( 1+u\frac{1}{\kappa}P \big) +\big( 1-(1-u)\frac{1}{\kappa}P\big) \otimes p_\mu }{1\otimes 1 +u(1-u)\left(\frac{1}{\kappa}\right)^2 P\otimes P}, \\
\Delta^{\mathcal{F}_{R,u}} D = \left( 1\otimes 1+\frac{u(1-u)}{\kappa^2} P\otimes P\right) \left( D\otimes \frac{1}{1+\frac{u}{\kappa}P} +\frac{1}{1-\frac{1-u}{\kappa}P}\otimes D\right), \\
S^{\mathcal{F}_{R,u}} (p_\mu) = -\frac{p_\mu }{1-(1-2u)\frac{1}{\kappa}P}, \\
S^{\mathcal{F}_{R,u}} (D) = -\left(1-\frac{1-u}{\kappa}P\right) D \left(\frac{1 - \frac{1-2u}{\kappa}P}{1-\frac{1-u}{\kappa}P}\right).
\end{gather*}

A similar analysis as in Section~\ref{sec:GZ} for $\Delta^{\mathcal{F}_{{\rm GZ},u}}p_\mu$ leads to the conclusion that $\Delta^{\mathcal{F}_{{\rm GZ},u}}D = \Delta^{\mathcal{F}_{R,u}}D$.

\subsection{Noncommutative coordinates and realizations}\label{ssec:nccoorcob}

In general, we consider realizations of the form
\begin{gather*}
 \hat{x}_\mu = x_\alpha \varphi_{\alpha \mu}(p) +\chi (p).
\end{gather*}
We can obtain the appropriate realization via the twist as follows
\begin{gather*}
 \hat{x}_\mu = m\big( \mathcal{F}_{R,u}^{-1} (\triangleright \otimes 1) (x_\mu\otimes 1) \big)
 =\left( x_\mu +\frac{{\rm i}}{\kappa} v_\mu (1-u)D \right) \left( 1+\frac{u}{\kappa}P\right) +u(1-u) \frac{{\rm i}}{\kappa^2}v_\mu P.
\end{gather*}

\subsection{Star product}\label{ssec:starcob}

Using the above realization of $\hat{x}_\mu$~\cite{mercati}, we get
\begin{gather*}
{\rm e}^{{\rm i}kx}*{\rm e}^{{\rm i}qx} = {\rm e}^{{\rm i}\mathcal{D}_\mu (u;k,q)x_\mu +{\rm i}\mathcal{G}(u;k,q)} = {\rm e}^{{\rm i}\mathcal{D}_\mu (u;k,q)x_\mu} \frac{1}{1+\frac{u(1-u)}{\kappa^2}(v\cdot k) (v\cdot q)},
\end{gather*}
where $k$ and $q$ belong to the $n$-dimensional Minkowski spacetime $\mathcal{M}_{1,n-1}$ and where
\begin{gather*}
\mathcal{D}_\mu (u;k,q)= \frac{k_\mu \big( 1+\frac{u}{\kappa}(v\cdot q)\big)+ \big(1-\frac{1-u}{\kappa}(v\cdot k)\big)q_\mu }{1+\frac{u(1-u)}{\kappa^2}(v\cdot k) (v\cdot q)}
\end{gather*}
as in equation~(\ref{eq:Dmu}), and finally
\begin{gather*}
\mathcal{G}(u;k,q) ={\rm i}\ln\left(1+\frac{u(1-u)}{\kappa^2}(v\cdot k) (v\cdot q)\right).
\end{gather*}
\begin{rem} Note that the corresponding quantum $R$-matrix is given by
\begin{gather*}
\mathcal{R}_{R,u}= \mathcal{F}_{R,u}^{21}\mathcal{F}_{R,u}^{-1} = \exp(u(PD\otimes 1 + 1\otimes PD)\mathcal{R}_0\exp(-u(PD\otimes 1+ 1\otimes PD)),
\end{gather*}
where
\begin{gather*}
\mathcal{R}_0 = \exp\left(-D\otimes\ln\left(1 - \frac{P}{\kappa}\right)\right)\exp\left(\ln\left(1-\frac{P}{\kappa}\right)\otimes D\right)\\
\hphantom{\mathcal{R}_0}{}
= \sum_{k,l = 0}^\infty \binom {-D}{l} \left(\frac{-P}{\kappa}\right)^k\otimes \left(\frac{-P}{\kappa}\right)^l \binom{D}{k}.
\end{gather*}
\end{rem}

Both twists, $\mathcal{F}_{{\rm GZ},u}^{-1}$ and $\mathcal{F}_{R,u}^{-1}$, lead to the same Hopf algebra, the same realizations of noncommutative coordinates $\hat{x}_\mu$ and likewise for the star product ${\rm e}^{{\rm i}kx}*{\rm e}^{{\rm i}qx}$. This suggests that there must be a close relation between the two twists, $\mathcal{F}_{{\rm GZ},u}^{-1}$ and $\mathcal{F}_{R,u}^{-1}$. In the next section, we present a proof that indeed $\mathcal{F}_{{\rm GZ},u}^{-1}=\mathcal{F}_{R,u}^{-1}$.

\section{Proofs of the equality of the two twists}\label{sec:equality}

\subsection[Differentiation with respect to parameter $u$]{Differentiation with respect to parameter $\boldsymbol{u}$}\label{ssec:diff}

Differentiating $\mathcal{F}_{R,u}^{-1}$ from equation~(\ref{eq:FRuexp}) with respect to the parameter $u$ gives
\begin{gather}\label{eq:kdduFRu}
\kappa \frac{{\rm d} \mathcal{F}_{R,u}^{-1}}{{\rm d}u} = (P\otimes D+D\otimes P) \mathcal{F}_{R,u}^{-1} +\big[PD\otimes 1+1\otimes PD,\mathcal{F}_{R,u}^{-1}\big].
\end{gather}
Differentiating $\mathcal{F}_{{\rm GZ},u}^{-1}$ from equation~(\ref{eq:FGZu}) with respect to $u$ gives
\begin{gather}
\kappa \frac{{\rm d} \mathcal{F}_{{\rm GZ},u}^{-1}}{{\rm d}u} =(P\otimes D +D\otimes P) \mathcal{F}_{{\rm GZ},u}^{-1}\nonumber\\
\hphantom{\kappa \frac{{\rm d} \mathcal{F}_{{\rm GZ},u}^{-1}}{{\rm d}u} =}{} +(P\otimes 1-1\otimes P) \sum_{k,l=0}^{\infty} \frac{-k+l}{\kappa^{k+l}} (u-1)^k P^k \binom D l \otimes (uP)^l \binom{D}{k}.\label{eq:kdduFGZu}
\end{gather}
Using the commutation relations
\begin{gather*}
\left[ PD, P^k\binom D l \right]=(l-k) P^{k+1} \binom D l
\end{gather*}
and
\begin{gather*}
\left[PD, P^l \binom D k \right]=(k-l)P^{l+1} \binom{D}{k},
\end{gather*}
we find that the right-hand sides of~(\ref{eq:kdduFRu}) and of~(\ref{eq:kdduFGZu}) agree,
\begin{gather*}
{\rm r.h.s.} = (P\otimes D+D\otimes P) \mathcal{F}_{{\rm GZ},u}^{-1} +\big[PD\otimes 1+1 \otimes PD, \mathcal{F}_{{\rm GZ},u}^{-1}\big].
\end{gather*}
This shows that $\mathcal{F}_{R,u}^{-1}$ and $\mathcal{F}_{{\rm GZ},u}^{-1}$ as functions of the parameter $u$ satisfy the same ordinary differential equation, while the initial conditions agree. Indeed, at $u=0$,
\begin{gather*}
\mathcal{F}_{R,u=0}^{-1}=\mathcal{F}_0^{-1}=\mathcal{F}_{{\rm GZ},u=0}^{-1}.
\end{gather*}
Therefore $\mathcal{F}_{R,u}^{-1}\equiv \mathcal{F}_{{\rm GZ},u}^{-1}$.

\subsection{Proof of the equality of the two twists}
\label{ssec:proofstar}

In the following proposition we state the conditions under which the two twists are equal, along with a simple proof.

\begin{proposition}Let $\mathcal{P}$ be the Poincar\'{e} Weyl algebra generated with momenta $p_\mu$, Lorentz gene\-rators $M_{\mu\nu}$ and dilatation $D$. Two twists $\mathcal{F}_1\in \mathcal{U}(\mathcal{P})\otimes \mathcal{U}(\mathcal{P})$ and $\mathcal{F}_2\in \mathcal{U}(\mathcal{P})\otimes \mathcal{U}(\mathcal{P})$ are identical if all the star products are identical, i.e., for all $f$ and $g$ in the Minkowski space time algebra,
\begin{gather*}
f*g= m\big( \mathcal{F}_1^{-1} (\triangleright \otimes \triangleright)(f\otimes g)\big) =m\big( \mathcal{F}_2^{-1} (\triangleright \otimes \triangleright)(f\otimes g)\big).
\end{gather*}\end{proposition}

\begin{proof} If all star products are the same, $\mathcal{F}_1^{-1}$ and $\mathcal{F}_2^{-1}$ could differ by an element in the right ideal $\mathcal{J}_0$ generated by the elements $(x_\mu \otimes 1- 1\otimes x_\mu)$ for all $\mu$~\cite{kov2,rina-PLA2013, IJMPA2014}. However, $\mathcal{J}_0\cap \mathcal{U}(\mathcal{P})\otimes \mathcal{U}(\mathcal{P})=0$, hence $\mathcal{F}_1=\mathcal{F}_2$.
\end{proof}

Since we already proved that the twists $\mathcal{F}_{R,u}$ and $\mathcal{F}_{{\rm GZ},u}$ give the same star products ${\rm e}^{{\rm i}kx} * {\rm e}^{{\rm i}qx}$, the twists $\mathcal{F}_{R,u}$ and $\mathcal{F}_{{\rm GZ},u}$ must be identical. Moreover, we have proved that the noncommutative coordinates $\hat{x}_\mu$ and twisted coproducts $\Delta p_\mu$ and $\Delta D$ from both twists are identical. Since $\mathcal{F}_{R,u}$ satisfies the normalization and 2-cocycle conditions, $\mathcal{F}_{{\rm GZ},u}$ also satisfies them.

\section{Conclusion}\label{sec:concl}

We have constructed a 1-parameter family $\mathcal{F}_{{\rm GZ},u}$~(\ref{eq:FGZu}) of Jordanian twists that interpolates between the simple Jordanian twists $\mathcal{F}_0$ and $\mathcal{F}_1$ defined in equation~(\ref{eq:F0F1}). We explicitly proved that $\mathcal{F}_{{\rm GZ},u}^{-1}$ satisfies the 2-cocycle condition~(\ref{eq:cocycFm}). For $u=\frac{1}{2}$, $\mathcal{F}_{{\rm GZ},u=\frac{1}{2}}$ coincides with $\mathcal{F}_{{\rm GZ}}$~\cite{GZ}. We have calculated the corresponding star product ${\rm e}^{{\rm i}kx}* {\rm e}^{{\rm i}qx}$~(\ref{eq:eikxstar}) and the corresponding deformed Hopf algebra structure. In Section~\ref{sec:cob}, we have presented another interpolation between Jordanian twists cohomologous to $\mathcal{F}_0$ via a 1-cochain depending on $u$~\cite{cobtw}. It is pointed out that $\mathcal{F}_{{\rm GZ},u}^{-1}$ and $\mathcal{F}_{R,u}^{-1}$ generate the same star product and the same deformed Hopf algebra. In Section~\ref{sec:equality}, a new result is presented that $\mathcal{F}^{-1}_{{\rm GZ},u} = \mathcal{F}^{-1}_{R,u}$, implying that $\mathcal{F}_{{\rm GZ},u}$ can be written in the form of a product of three exponential factors. Twist $\mathcal{F}_{R,u}$ automatically satisfies the 2-cocycle condition as it is obtained from a simple Jordanian twist by twisting by a 1-cochain~\cite{majid}. We note that for the twist $\mathcal{F}^{-1}_{{\rm GZ}}$~\cite{GZ}, the star product, an explicit form of the twist $\mathcal{F}_{{\rm GZ}}$ and the deformed Hopf algebra structure, were not known in the literature so far. Jordanian twists have been of interest in the recent literature~\cite{DJP_SIGMA2012,tongeren, kmps,PRD19,PRD-conformal,pv}. We note that our results could be useful in future applications of Jordanian twists.

\appendix
\section{Appendix}\label{appendixA}

\begin{lemma} If $x$, $y$, $z$ are mutually commuting variables, and $k$, $l$, $A$, $C$ with $A\leq k$, $C\leq l$ nonnegative integers, then
\begin{gather}
 \binom{x}{l-C}\sum_{k_1 = k-A}^k \binom{k_1}{k-A}\binom{z}{k_1} \binom{x+y-k-l+k_1+C}{C}\binom{y}{k-k_1}\nonumber\\
 \qquad {} =\binom{z}{k-A}\sum_{l_1 = l-C}^l \binom{l_1}{l-C}\binom{x}{l_1}\binom{y+z-k-l+l_1+A}{A}\binom{y}{l-l_1}.\label{eq:bigident}
\end{gather}
\end{lemma}

\begin{proof} To make the proof more transparent, we make a change of summation indices
\begin{gather*}
 i = k_1 - k + A,\qquad j= l_1- l + C,
\end{gather*}
hence $k_1 = i + k - A$ and $l_1 = j + l - C$, to restate equation~(\ref{eq:bigident}) as
\begin{gather}
\binom{x}{l-C}\sum_{i = 0}^A \binom{i+k-A}{k-A}\binom{z}{i+k-A}\binom{x+y-l+i-A+C}{C}\binom{y}{A-i}\nonumber\\
\qquad{} = \binom{z}{k-A}\sum_{j = 0}^C \binom{j+l-C}{l-C}\binom{x}{j+l-C}\binom{y+z-l+j-C+A}{A}\binom{y}{C-j}.\label{eq:befm}
\end{gather}

We remind the reader of the simple identity
\begin{gather*}
 \binom{r}{s}\binom{w}{r} = \binom{w}{s}\binom{w-s}{r-s},
\end{gather*}
which we apply in~(\ref{eq:befm}) for $w = z$ on the left and for $w = x$ on the right, to obtain an equivalent statement,
\begin{gather}
 \binom{x}{l-C}\binom{z}{k-A}\sum_{i = 0}^A\binom{z-k+A}{i} \binom{x+y-l+i-A+C}{C}\binom{y}{A-i}\nonumber\\
 \qquad{} =\binom{z}{k-A}\binom{x}{l-C}\sum_{j = 0}^C\binom{x}{j+l-C}\binom{y+z-k+j-C+A}{A}\binom{y}{C-j}.\label{eq:final}
\end{gather}
We expand
\begin{gather*}
\binom{x+y-l+i-A + C}{C} = \sum_{j = 0}^C\binom{x-l+C}{j}\binom{y+i-A}{C-j}
\end{gather*}
on the left-hand side, and
\begin{gather*}
\binom{y+z-k+j-C+A}{A} = \sum_{i = 0}^A\binom{z-k+A}{i}\binom{y+j-C}{A-i}
\end{gather*}
on the right-hand side of~(\ref{eq:final}). Now both sides involve double summation over $i$ and $j$. For each fixed pair $(i,j)$, compare the corresponding summands on the two sides. The factors involving $x$ and $z$ are identical on both sides. It remains to check that the factors involving $y$ agree. Indeed, by definition,
\begin{gather*}
 \binom{y + i - A}{C-j}\binom{y}{A-i} = \frac{y(y-1)\cdots (y+i-A+j-C+1)}{(C-j)! (A-i)!} = \binom{y + j - C}{A-i}\binom{y}{C-j}.\!\!\!\!\tag*{\qed}
\end{gather*}\renewcommand{\qed}{}
\end{proof}

\subsection*{Acknowledgements}
We thank Anna Pacho{\l} for useful discussions. Z.\v{S}.~has been partly supported by the Croatian Science Foundation under the Project ``New Geometries for Gravity and Spacetime'' (IP-2018-01-7615) and by the grant 18-00496S of the Czech Science Foundation.

\pdfbookmark[1]{References}{ref}
\LastPageEnding

\end{document}